\documentclass[journal,10pt]{IEEEtran}

\usepackage{amsmath,amssymb,amsthm}
\usepackage{graphicx}
\usepackage{booktabs}
\usepackage{array}
\usepackage{tabularx}
\usepackage{cite}
\usepackage{url}
\usepackage{balance}
\usepackage{enumitem}

\newtheorem{theorem}{Theorem}
\newtheorem{proposition}{Proposition}
\newtheorem{corollary}{Corollary}

\newcommand{\R}{\mathbb{R}}

\newcommand{\calP}{\mathcal{P}}

\newcommand{\argmax}{\operatorname*{arg\,max}}

\newcolumntype{Y}{>{\raggedright\arraybackslash}X}

\title{A Graph-Based Stackelberg Security Game for Trustworthy 6G Disaggregated Architecture}

\author{Lav R. Varshney and Xinbo Wu%
\thanks{L. R. Varshney is with the AI Innovation Institute, Stony Brook University, Stony Brook, NY 11794 USA (e-mail:lav.varshney@stonybrook.edu). X. Wu was with the University of Illinois Urbana-Champaign.}
\thanks{ChatGPT 5.6-Sol was used to support experiment implementation, mathematical development, and writing.}}

\begin{document}
\maketitle

\begin{abstract}
Open, cloud-native, and AI-enabled 6G architectures make network functions easier to deploy, observe, and replace, often implemented with separate software modules.
However, each explicit interface can also create an entry point, trust transition, and lateral-movement route from a cybersecurity perspective. This paper formulates architectural disaggregation as a graph-based Stackelberg design game for 6G system security. A deployment graph specifies candidate 6G system seams, while an architecture-dependent directed attack graph specifies externally reachable functions and multistage paths. The defender commits to seams and allocates hardening and monitoring; an informed attacker then chooses a path and effort. Because edge traversal and evasion factors compose multiplicatively, the attacker's path response becomes a shortest-path problem after a logarithmic transformation. Convex effort costs admit a conjugate representation, and, for each fixed architecture, the defender's control problem is a convex exponential-epigraph program solvable by path generation. We derive an architecture-comparison criterion, a threshold for AI transducers within 6G systems that jointly add visibility and exposure, and comparative statics for attacker entry and bypass innovation. Experiments with stylized standards-anchored O-RAN/6G systems over RAN, RIC, core, cloud, and edge-AI functions show that neither full integration nor maximal disaggregation is generally best. Uncontrolled seams can increase risk, whereas selectively monitored and hardened seams can improve the Stackelberg objective. We find a boundary is worthwhile only when its detection and containment gains exceed its added attack opportunities and coordination cost.
\end{abstract}


\section{Introduction}
\IEEEPARstart{T}{he} July 2026 \emph{Call to Action for 6G Leadership and Security} asks partner governments to advance 6G systems that are open, interoperable, secure, resilient, AI-supported, and resource-efficient \cite{ntia2026call}. A preceding joint statement likewise endorsed 6G technology that is secure, open, interoperable, reliable, and resilient by design \cite{state2024principles}. The interoperability point implies softwarization of 6G system design, e.g.\ as in O-RAN frameworks, should enable modules from different vendors to substitute for one another. With standardized interfaces, functions can be independently deployed, inspected, and replaced.  Yet, introducing such interfaces rather than having single-vendor monolithic design creates attack surfaces---endpoints, credentials, software dependencies, and trust transitions that an attacker can exploit.  At the same time, such interfaces also enhance visibility that may be used to protect against attacks.  We consider the tension between monolithic and modular architecture from a security perspective. 

The same tension appears not just in policy documents, but also in technical 6G roadmaps. ITU-R expects IMT-2030 to be secure by design while using standardized interoperable interfaces and integrating distributed computing and AI capabilities \cite{itu2023imt2030}. The Next Generation Mobile Networks (NGMN) alliance lists modular delivery and trustworthiness as key architectural principles \cite{ngmn2025architecture}. In O-RAN, open fronthaul, E2, A1, O1, O2, F1, R1, and service-management interfaces support independently implemented radio, control, orchestration, and application functions, but also create interface-specific threats and controls \cite{polese2023understanding,porambage2024security,baguer2024attacking,oran2026architecture,oran2026security,oran2026threat}. 

Introducing AI into the system influences both sides of the tradeoff. Edge intelligence and cloud-native services can improve visibility and control \cite{zeb2021edge,letaief2022edgeai,soltani2025intelligent}. EdgeRIC, for example, places a real-time controller near the distributed unit (DU), decouples it from the RAN stack, and enables sub-millisecond AI-in-the-loop control \cite{ko2024edgeric}. Such transducers expose fine-grained telemetry and localize anomalous behavior.  At the same time, they also add high-dimensional vulnerabilities. From a security perspective, therefore, one might wonder whether a particular boundary, together with the controls placed there, improves the strategic balance against an attacker who chooses both route and effort.

Attack graphs provide a natural representation of multistage cyberattacks/compromises \cite{phillips1998graph,sheyner2002attackgraphs,ammann2002scalable}, and Stackelberg games are natural when a defender commits to an architecture before an attacker selects a target or plan \cite{alpcan2010network,tambe2011security}. There is past work on strategic defense over attack graphs, including optimizing honeypots, hardening portfolios, and Bayesian control selection against adaptive attackers \cite{durkota2015optimal,durkota2016case,durkota2019hardening,zhang2021bayesian}. Those formulations largely take the underlying architecture as given. Here we ask the question one level up: which interfaces and trust transitions should exist at all? 
Conversely, architectural discussions often count interfaces without allowing the attacker to reroute or scale effort. Here the architecture itself is part of the leader's commitment.  That is, our contribution is not the use of an attack graph or Stackelberg formulation, but rather their coupling in the architectural disaggregation decision.

The main contributions of the paper are as follows.
\begin{enumerate}
\item We formulate a two-level 6G architecture model: a deployment graph determines modules and coordination costs, while an architecture-dependent directed attack graph determines entry points, traversal opportunities, monitoring, and hardening.
\item We solve the attacker's path-and-effort best response. Multiplicative path yield becomes a shortest-path problem, and convex effort costs yield a closed-form solution through convex conjugacy. The quadratic case yields an explicit Stackelberg loss term for cyberattack (cyber-loss).
\item We prove that, for a fixed architecture, optimal hardening and monitoring form a convex program. We derive a path-generation implementation and an architecture-comparison criterion stating when the security gain of a proposed boundary pays for its added cost.
\item We obtain comparative statics giving some policy-relevant insights: an AI transducer is beneficial when its visibility and hardening overcome its added exposure, improved defense contracts the set of profitable attacker types, and sufficiently protected ordinary routes shift the residual threat toward actors that can develop novel bypasses.
\item We consider modular partitioning of a stylized standards-anchored O-RAN system architecture, finding the seams to introduce that yield the best intermediate level of modularity from an architectural security perspective.
\end{enumerate}
We aim to show the mathematical structure of the monolithic--modular architecture tradeoff in trustworthy 6G design.

Note that unlike spectral or multilevel graph partitioning, which seek low-cut or high-affinity vertex groups under fixed edge weights, our partition is induced by technologically admissible architectural seams and is evaluated after an informed attacker reoptimizes its target, route, and effort. It is an adversarial network design problem rather than a community detection problem \cite{RoddenberrySWS2020}.  The formulation may be more generally applicable for other engineering system design problems. 

\section{Standards-Inspired
Architecture and Threat Model}

This section develops the mathematical framework for system architecture and threat model.  It is grounded in 6G and particularly O-RAN frameworks, as discussed.

\subsection{Background: 6G, O-RAN, and zero-trust controls}
\label{sec:oran}

\begin{table*}[t]
\caption{Candidate O-RAN Architectural Seams}
\label{tab:seams}
\centering
\small
\renewcommand{\arraystretch}{1.12}
\begin{tabularx}{\textwidth}{@{}p{0.055\textwidth}p{0.175\textwidth}YY@{}}
\toprule
Symbol & Boundary & Architectural role & Exposure and control interpretation \\
\midrule
FH & O-RU--O-DU & Open-fronthaul split between radio and lower-layer baseband functions. & Adds a remotely reachable fronthaul surface; enables protocol-specific monitoring and hardening. \\
RT & O-DU--EdgeRIC & Real-time AI control placed near the DU. & Adds a high-leverage controller/API surface; supplies strong local telemetry. \\
F1 & O-DU--O-CU-CP/UP & Distributed--centralized RAN split over F1-C/F1-U. & Adds control- and user-plane seams with moderate coordination and monitoring value. \\
RIC & O-CU/DU--near-RT RIC--xApp & Makes E2 control and xApp lifecycle boundaries explicit. & Adds application and control entry points; offers high telemetry value but substantial software exposure. \\
SMO & RIC--SMO--rApp/O-Cloud & Represents A1, O1, O2, R1, model, and orchestration interactions. & Creates management and model-supply paths; enables policy and lifecycle monitoring. \\
N2N3 & RAN--5GC & Separates RAN control and user planes from AMF/UPF. & Adds N2/N3 trust transitions with relatively mature control mechanisms. \\
SBA & 5GC service functions & Makes service-based core functions independently addressable. & Adds API, discovery, authorization, and exposure paths; enables service-level monitoring. \\
EDGE & UPF--MEC/AI service & Separates user-plane anchoring from edge applications and AI inference. & Adds application and model exposure; supports localized flow and application monitoring. \\
\bottomrule
\end{tabularx}
\end{table*}

Practical experimentation of attacks on O-RAN has provided concrete descriptions of threats \cite{baguer2024attacking}.  The intelligent control literature treats the RIC simultaneously as a security risk and a defensive opportunity \cite{soltani2025intelligent}. Notably, the O-RAN WG11 specifications assign security requirements and risk analysis to these architectural elements \cite{oran2026security,oran2026threat}. Zero-trust guidance emphasizes explicit authentication and authorization, least privilege, and continuous assessment rather than implicit trust based on location \cite{nist2020zta,oran2024zta}. In our formulation, these controls map naturally to the notion that hardening reduces successful traversal and monitoring reduces undetected traversal in attack paths. A logical boundary provides no security by itself; it has value only insofar as deployable controls change the attacker's best route.

A stylized O-RAN architecture has 15 functions: O-RU, O-DU, O-CU-CP, O-CU-UP, a real-time EdgeRIC, near-RT RIC, xApp, SMO/non-RT RIC, rApp/model manager, O-Cloud, AMF, SMF, UPF, an aggregated NRF/NEF/AUSF service, and an MEC/AI service, see \cite{threegpp23501,oran2026architecture}.  As shown in Table~\ref{tab:seams}, there are 8 key boundaries where modularization can be performed: open fronthaul, real-time RIC/DU, F1, E2/xApp, A1/O1/O2/rApp, RAN/core, service-based core, and UPF/edge-AI boundaries.  The table also gives much more detail on why these are potential points of modularization, and what kinds of visibility/control can be imposed.

\subsection{Deployment graph and candidate seams}
Let $G_F=(V,E_F)$ be an undirected functional-dependency graph. A vertex is a network function or control service, and an edge means the endpoints exchange user-plane, control-plane, policy, management, telemetry, or model information. An architecture is represented by a set of candidate seams $\sigma\subseteq\mathcal{B}$, where $\mathcal{B}$ is the finite set of candidate architectural seams, the collection of disaggregation decisions available to the defender. Removing the edges associated with $\sigma$ creates a partition $\pi_\sigma$ of the system into $K(\sigma)$ modules. A monolithic architecture has $\sigma=\emptyset$ and $K=1$.

There are some initial costs with different architectures independent of security considerations (non-cyber cost):
\begin{equation}
 C_{\rm arch}(\sigma)
 =c_M\bigl(K(\sigma)-1\bigr)
 +\sum_{b\in\sigma}\kappa_b ,
 \label{eq:archcost}
\end{equation}
where $c_M$ is lifecycle overhead for an independently deployed module and $\kappa_b$ is interface signaling, serialization, synchronization, identity management, interoperability testing, and certification. Equation~\eqref{eq:archcost} is meant to be interpretable, but a more detailed implementation can replace $\kappa_b$ by measured latency, message rate, or engineering labor without changing the security game.

This architecture cost is meant to be a baseline associated with modularity and is related to the thermodynamics  of modularity. Boyd, Mandal, and Crutchfield showed that localized operations can irreversibly discard globally useful correlations, producing a modularity cost beyond ordinary Landauer erasure \cite{boyd2018thermodynamics}. As such there is an operational cost of coordination information made unavailable when jointly useful states are separated across a boundary. Our $\kappa_b$ is similar from the need for signaling, synchronization, and interoperability testing/certification.

\subsection{Architecture-dependent attack graph}
Now consider attacks.  For architecture $\sigma$, let
\begin{equation}
G_A(\sigma)=\bigl(\{s\}\cup V,E_A(\sigma)\bigr)
\end{equation}
be a directed attack graph, with external source $s$. A directed edge $e=(u,v)$ means that external access to $v$ or compromise of $u$ can enable compromise, control manipulation, or evasion at $v$. Selecting a seam can alter $G_A$ in two ways: it can increase the baseline attractiveness of traversal edges that cross a newly explicit interface; and it can create a new source edge $s\to v$ representing an API, management endpoint, application package, or supply-chain entry point.
Thus attack surface is represented by number and quality of feasible paths rather than by an arbitrary penalty proportional only to interface count.

Each edge has a baseline exploitation/traversal factor $\rho_e^0\in(0,1]$ and a baseline evasion factor $q_e^0\in(0,1]$. A controllable edge is assigned to a security site $j(e)$. The defender allocates hardening $h_j\geq0$ and monitoring $n_j\geq0$. The resulting undetected traversal factor is
\begin{equation}
 r_e(h,n)=\rho_e^0 q_e^0
 \exp\{-\gamma_j h_j-I_j n_j\},
 \label{eq:edgefactor}
\end{equation}
where $\gamma_j$ is hardening effectiveness and $I_j$ is detector effectiveness. Uncontrolled edges use $h_j=n_j=0$. Hardening can summarize authentication, least privilege, protocol validation, isolation, software assurance, or patching. Monitoring can summarize an interface-level detector operated at a specified false-alarm constraint. Under repeated conditionally independent observations, an exponential miss model is consistent with standard detection-error exponents \cite{kay1998detection,chernoff1952measure}; operationally, $I_j$ may instead be fitted directly from a detector's receiver operating characteristic.

To emphasize the point, we use an exponential undetected-traversal model for two related reasons. First, under fixed false-alarm operation, the miss probability of a statistically regular interface detector often decays exponentially with effective sample size or monitoring duration; $I_j$ represents the corresponding empirical or asymptotic information rate. Second, we model hardening through constant proportional reduction in residual traversal effectiveness, yielding the proportional-hazards form $\rho_e(h_e) = \rho_e^0e^{-\gamma_eh_e}$. Hence the negative logarithm of undetected traversal effectiveness is additive in hardening, monitoring, and attack-path composition. The exponential form should be thought of as a tractable form from basic statistical signal processing rather than a universal empirical finding.  Note that irreducible-risk floors and detector-response delays can be incorporated without changing the basic Stackelberg framework.

The product model assumes conditionally independent edge events along a path. This need not be literal independence: $r_e$ can be interpreted as a calibrated conditional factor given arrival at edge $e$. Correlated evidence can be incorporated by replacing the sum of edge exponents with an empirically estimated path-level exponent.

\subsection{Path yield}
Let $\calP(\sigma)$ be the set of directed paths from $s$ to attack-relevant targets. If path $p$ ends at target $t(p)$ with consequence $V_{t(p)}>0$, define
\begin{equation}
 \Theta_p(\sigma,h,n)
 =V_{t(p)}\prod_{e\in p}r_e(h,n).
 \label{eq:pathyield}
\end{equation}
The quantity $\Theta_p$ is the expected undetected mission yield per unit attacker effort. The architecture's maximum attack yield is
\begin{equation}
 \Theta(\sigma,h,n)=\max_{p\in\calP(\sigma)}\Theta_p(\sigma,h,n).
 \label{eq:theta}
\end{equation}
A high-value target can therefore be dominated by a lower-value target that is much easier to reach and evade detection.

This model is intended to capture several features particularly salient in disaggregated 6G systems. First, architectural boundaries are not merely abstract graph cuts: interfaces such as open fronthaul, E2, A1/O1/O2, service-based core APIs, and edge-AI boundaries create explicit trust transitions that can simultaneously introduce new remotely reachable attack surfaces and new locations for authentication, isolation, telemetry, and monitoring. Second, attacks on such systems are naturally multistage, since reaching a high-consequence control or forwarding function may require a sequence of entry, lateral movement, privilege transition, and evasion steps. This motivates representing path effectiveness multiplicatively, while hardening and monitoring act locally on the corresponding trust transitions. The Stackelberg timing we now develop reflects that standards choices, module boundaries, and deployed controls are committed on a substantially slower timescale than an attacker choosing a route and campaign effort, allowing the attacker to adapt to the exposed architecture.

\section{The Graph-Based Stackelberg Game}
Consider the Stackelberg game shown in Fig.~\ref{fig:game}.
The defender is the leader and publicly commits to
\begin{equation}
 d=(\sigma,h,n).
\end{equation}
The attacker observes $d$ and chooses a path $p\in\calP(\sigma)$ and a nonnegative effort $a$. Its utility is
\begin{equation}
 U_A(p,a;d)=a\Theta_p(d)-g(a),
 \label{eq:attackerutility}
\end{equation}
where $g$ is closed, strictly convex, increasing on $\R_+$, satisfies $g(0)=0$, and grows superlinearly. The defender's expected mission loss is $a\Theta_p$ and its total objective is
\begin{equation}
 J_D(d;p,a)=C_{\rm arch}(\sigma)+c_h^T h+c_n^T n
 +\lambda a\Theta_p(d),
 \label{eq:defenderraw}
\end{equation}
where $\lambda$ converts the consequence of cyber attack into the same normalized units as architecture and control cost, which themselves have costs $c_h$ and $c_n$.

\begin{figure}
\centering
\includegraphics[width=\columnwidth]{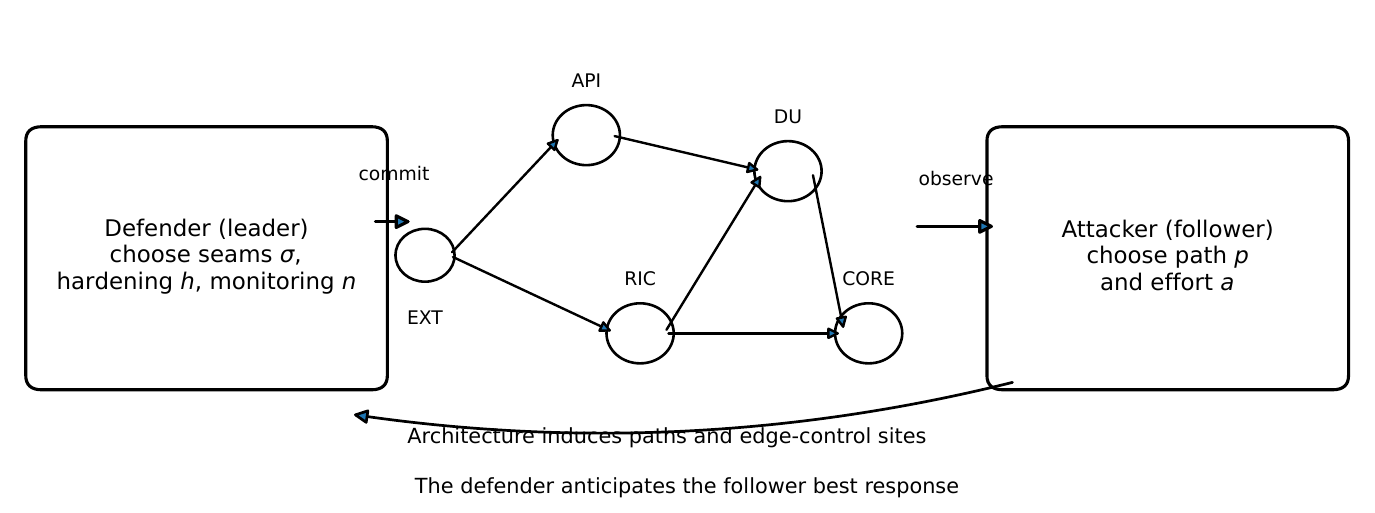}
\caption{Stackelberg timing. The defender changes both the feasible attack graph and the sites at which controls may be placed.}
\label{fig:game}
\end{figure}

\subsection{Attacker best response}
Define the restricted convex conjugate \cite{rockafellar1970convex}
\begin{equation}
 g_+^*(y)=\sup_{a\geq0}\{ay-g(a)\}.
\end{equation}

\begin{proposition}[Path and effort best response]
For fixed defender design $d$, the attacker chooses a path in $\argmax_p\Theta_p(d)$ and obtains
\begin{equation}
 U_A^*(d)=g_+^*(\Theta(d)).
 \label{eq:conjugate}
\end{equation}
If $g$ is differentiable and the optimum is interior, $g'(a^*)=\Theta(d)$.
\end{proposition}
\begin{IEEEproof}
For each fixed $a\geq0$, the term $-g(a)$ is independent of $p$, so maximizing over paths gives $a\Theta(d)-g(a)$. Maximizing this scalar expression over $a\geq0$ is precisely $g_+^*(\Theta(d))$. Strict convexity of $g$ gives uniqueness of an interior effort response and the first-order condition.
\end{IEEEproof}

Consider an example.  For the quadratic cost
\begin{equation}
 g(a)=\tfrac{c_A}{2}a^2,
 \label{eq:quadcost}
\end{equation}
we obtain
\begin{equation}
 a^*(d)=\tfrac{\Theta(d)}{c_A},\quad
 U_A^*(d)=\tfrac{\Theta(d)^2}{2c_A},
 \label{eq:quadresponse}
\end{equation}
while the defender's equilibrium mission loss is
\begin{equation}
 L_{\rm cyber}(d)=a^*\Theta(d)=\tfrac{\Theta(d)^2}{c_A}.
 \label{eq:cyberloss}
\end{equation}
The Stackelberg design problem is therefore
\begin{equation}
 \min_{\sigma,h,n}
 \;C_{\rm arch}(\sigma)+c_h^T h+c_n^T n
 +\beta\Theta(\sigma,h,n)^2,
 \quad \beta=\tfrac{\lambda}{c_A}.
 \label{eq:leaderproblem}
\end{equation}
The attacker is strategic in two ways: it shifts to the most attractive route after every architectural change, and it increases effort when the architecture offers a larger return.

\subsection{Shortest-path representation}
Let
\begin{equation}
 \ell_e(h,n)=-\log r_e(h,n)\geq0.
\end{equation}
For a target $v$, let $d_{h,n}(s,v)$ be the shortest directed-path length from $s$ to $v$ under edge lengths $\ell_e$.

\begin{theorem}[Best attack as a shortest path]\label{thm:shortestpath}
For any fixed architecture and controls,
\begin{equation}
 \log\Theta(\sigma,h,n)
 =\max_{v\in V}\left\{\log V_v-d_{h,n}(s,v)\right\}.
 \label{eq:shortestpath}
\end{equation}
Any shortest path to a maximizing target is an attacker best response.
\end{theorem}
\begin{IEEEproof}
For a path $p$ ending at $v$, taking the log of \eqref{eq:pathyield} gives
$\log\Theta_p=\log V_v-\sum_{e\in p}\ell_e$.
For each fixed $v$, maximizing over paths is equivalent to minimizing the path length, which gives $d_{h,n}(s,v)$. Maximizing over targets proves \eqref{eq:shortestpath}.
\end{IEEEproof}

The theorem gives a straightforward attack-path calculation. Nonnegative edge lengths permit the use of Dijkstra's algorithm \cite{dijkstra1959note}.  For directed acyclic graphs (DAGs), as would arise for 6G systems, there is a linear-time topological-order solution.

\section{Convex Security Control and Architectural Criteria}
Consider insights into interface control, e.g.\ for the quadratic cost setting.

\subsection{Fixed-architecture convex program}
Fix $\sigma$. Stack the controls as $x=(h^T,n^T)^T$ and define
\begin{equation}
 b_p=\log V_{t(p)}+\sum_{e\in p}\log(\rho_e^0q_e^0),
\end{equation}
with row vector $A_p\geq0$ collecting the $\gamma_j$ and $I_j$ coefficients encountered on path $p$. Then
\begin{equation}
 \log\Theta_p=b_p-A_px,
 \qquad
 \log\Theta=\max_{p\in\calP}(b_p-A_px).
 \label{eq:logtheta}
\end{equation}
For box constraints $0\leq x\leq\bar x$, the fixed-architecture problem is
\begin{subequations}\label{eq:convexprogram}
\begin{align}
 \min_{x,z}\quad & c^Tx+\beta e^{2z}\label{eq:convexobj}\\
 \text{s.t.}\quad &z+A_px\geq b_p,\quad p\in\calP(\sigma),\label{eq:pathconstraints}\\
 &0\leq x\leq\bar x.\label{eq:box}
\end{align}
\end{subequations}

\begin{theorem}[Convexity and exactness]
Problem \eqref{eq:convexprogram} is a convex program and has the same optimum as the continuous inner minimization in \eqref{eq:leaderproblem} for fixed $\sigma$.
\end{theorem}
\begin{IEEEproof}
The objective is the sum of a linear function and the convex function $\beta e^{2z}$. Every path constraint is affine, and the box is convex. At any feasible $(x,z)$, \eqref{eq:pathconstraints} implies $z\geq\max_p(b_p-A_px)=\log\Theta(x)$. Since $e^{2z}$ is increasing, every optimum sets $z=\log\Theta(x)$, producing objective $c^Tx+\beta\Theta(x)^2$.
\end{IEEEproof}

The number of paths may be exponential even when the graph is sparse. The formulation nevertheless admits constraint generation. One would solve \eqref{eq:convexprogram} with a subset of path constraints, compute the best attack through \eqref{eq:shortestpath}, add the corresponding violated constraint, and repeat. Separation is therefore a shortest-path computation rather than path enumeration.

The KKT conditions also have an operational interpretation. Let $\mu_p\geq0$ be the multiplier on path $p$. Ignoring active box bounds,
\begin{equation}
 2\beta e^{2z}=\sum_p\mu_p,
 \qquad
 c_j=\sum_p\mu_p A_{pj}.
 \label{eq:kkt}
\end{equation}
Controls should be put in place where their hardening or detection effectiveness covers several co-dominant high-value paths. An interface that appears prominent but absent from active paths should receive no marginal investment.

\subsection{Robust and multi-campaign extensions}
The maximum-path game is appropriate when an actor concentrates a campaign on its best available route. One can also extend in two ways that maintain mathematical tractability.

First, baseline path yields can be uncertain. Let
\begin{equation}
 b_p(u)=\widehat b_p+d_p^Tu,
 \qquad u\in\mathcal{U},
 \label{eq:uncertainb}
\end{equation}
where $\mathcal{U}$ is a compact convex uncertainty set representing uncertain exploitability, detector calibration, or target consequence on the logarithmic scale.

\begin{proposition}[Robust fixed-architecture design]
The defender problem that protects against every $u\in\mathcal{U}$ is the convex program \eqref{eq:convexprogram} with each path constraint replaced by
\begin{equation}
 z+A_px\geq \widehat b_p+\sigma_{\mathcal U}(d_p),
 \label{eq:robustconstraint}
\end{equation}
where $\sigma_{\mathcal U}(d)=\sup_{u\in\mathcal U}d^Tu$ is the support function. For independent box uncertainty $|u_i|\leq\delta_i$, the additive margin is $\sum_i\delta_i|d_{pi}|$.
\end{proposition}
\begin{IEEEproof}
The robust counterpart of $z+A_px\geq b_p(u)$ requires the left side to exceed $\sup_{u\in\mathcal U}(\widehat b_p+d_p^Tu)$, which is \eqref{eq:robustconstraint}. The right side is a finite constant for every path, so convexity is unchanged.
\end{IEEEproof}

This result is useful when red teaming or incident data provide confidence intervals rather than precise edge probabilities. It makes conservative design transparent in the sense that uncertainty does not require a new risk metric; it only shifts each path's log-yield by its worst credible margin. A budgeted or ellipsoidal set can reduce the conservatism of independent intervals while retaining a support-function constraint \cite{rockafellar1970convex}.

Second, a capable adversary may run several campaigns rather than selecting one route. Let $a_p\geq0$ be effort allocated to campaign class $p$, with utility
\begin{equation}
 U_A(\mathbf a;d)=\sum_{p\in\mathcal C}a_p\Theta_p(d)
 -\tfrac{1}{2}\sum_{p\in\mathcal C}c_p a_p^2,
 \label{eq:multicampaignutility}
\end{equation}
where $\mathcal C$ contains operationally distinct campaigns rather than every overlapping logical path. Then $a_p^*=\Theta_p/c_p$ independently.

\begin{proposition}[Parallel-campaign convexity]
If defender loss is $\sum_p a_p^*\Theta_p$, then for fixed $\sigma$ the control problem is
\begin{equation}
 \min_{0\leq x\leq\bar x}
 c^Tx+\lambda\sum_{p\in\mathcal C}
 \tfrac{\exp\{2(b_p-A_px)\}}{c_p},
 \label{eq:multicampaignconvex}
\end{equation}
which is convex.
\end{proposition}
\begin{IEEEproof}
Each term is a positive multiple of the exponential of an affine function, hence convex; their sum plus a linear control cost is convex.
\end{IEEEproof}

The single-best-path model  is more conservative about strategic concentration but less aggressive about attack volume. Equation~\eqref{eq:multicampaignconvex} shows the main conclusion does not depend on a single attack. It also separates two policy regimes: when defenses leave one route dominant, the maximum-path Stackelberg game is natural, but when automation makes many low-cost campaigns simultaneously feasible, a parallel-campaign objective is more appropriate.

\subsection{Architecture comparison}
Let $d=(\sigma,h,n)$ and $d'=(\sigma',h',n')$ be any two fully specified designs, and let
\begin{equation}
 C_D(d)=C_{\rm arch}(\sigma)+c_h^Th+c_n^Tn.
\end{equation}

\begin{proposition}[When disaggregation is preferable]
Under quadratic attacker effort, design $d'$ is preferable to $d$ exactly when
\begin{equation}
 C_D(d')-C_D(d)
 <\beta\left[\Theta(d)^2-\Theta(d')^2\right].
 \label{eq:exactcriterion}
\end{equation}
\end{proposition}
\begin{IEEEproof}
Subtract the defender objectives in \eqref{eq:leaderproblem}; \eqref{eq:exactcriterion} is equivalent to $J_D(d')<J_D(d)$.
\end{IEEEproof}

Although direct, the proposition is operationally insightful.  A module boundary has no intrinsic sign. Rather, its value is determined by how it changes the maximum route after the attacker is allowed to switch paths.

For a local calculation, suppose a proposed boundary does not change the identity of the dominant path and multiplies its yield by
\begin{equation}
 m_b=\alpha_b\eta_b e^{-\gamma_bh_b-I_bn_b},
 \label{eq:multiplier}
\end{equation}
where $\alpha_b\geq1$ is added exposure, $\eta_b\leq1$ is containment or privilege reduction, and the exponential term captures hardening and detection. Then security risk falls if and only if $m_b<1$, and the boundary is net beneficial if and only if
\begin{equation}
 \Delta C_b<\beta\Theta^2(1-m_b^2).
 \label{eq:localcriterion}
\end{equation}
Here $\Delta C_b = C_D(d_b) - C_D(d)$ is the incremental defender-side cost of introducing and securing the proposed boundary $b$.
Note that this local criterion is not to be used after a path switch: the global program automatically handles switches.

\begin{corollary}[AI transducer threshold]
Suppose an AI transducer added at boundary $b$ creates exposure multiplier $\alpha_b$ but provides containment $\eta_b$, hardening $h_b$, and detector exponent $I_bn_b$. It reduces yield along the current dominant path if and only if
\begin{equation}
\gamma_bh_b+I_bn_b>\log(\alpha_b\eta_b).
 \label{eq:AIthreshold}
\end{equation}
It improves the defender objective only if \eqref{eq:localcriterion} also covers its implementation and coordination cost.
\end{corollary}
This corollary formalizes a central AI security question. Better anomaly classification alone is insufficient if the model-serving endpoint, data pipeline, or control API adds more exploitable opportunity than the detector removes.

\subsection{Attacker entry and offensive innovation}
Consider the attacker innovating to develop new attacks.  Let attacker type $t$ have effort parameter $c_t$, fixed development cost $F_t$, and expected attribution or response cost $A_t$. Its optimized payoff on existing routes is
\begin{equation}
 \Pi_t(d)=\tfrac{\nu_t^2\Theta(d)^2}{2c_t}-F_t-A_t,
 \label{eq:typepayoff}
\end{equation}
where $\nu_t$ scales operational value or capability.

\begin{proposition}[Contraction of the active attacker set]
If two defenses satisfy $\Theta(d')\leq\Theta(d)$ and all type parameters are fixed, then
\begin{equation}
 \{t:\Pi_t(d')>0\}\subseteq\{t:\Pi_t(d)>0\}.
 \label{eq:activecontraction}
\end{equation}
\end{proposition}
\begin{IEEEproof}
Each $\Pi_t$ is nondecreasing in $\Theta^2$, so any type unprofitable at $d$ remains unprofitable at lower-yield design $d'$.
\end{IEEEproof}

To distinguish attack volume from innovation, suppose type $t$ can pay $K_t$ to discover and execute a novel bypass route with yield $\Theta_t^{\rm new}>\Theta(d)$. The bypass is preferable to commodity routes precisely when
\begin{equation}
K_t<\tfrac{\nu_t^2\bigl[(\Theta_t^{\rm new})^2-\Theta(d)^2\bigr]}{2c_t}.
 \label{eq:bypass}
\end{equation}
As ordinary paths are hardened and monitored, this relative willingness to pay for a bypass increases. Defenses can therefore contract the ordinary attacker population while increasing the incremental value of changing the attack graph for a remaining high-capability actor. 

AI can act on both sides: secure development and automated patching lower baseline $\rho_e^0$ \cite{nist2022ssdf}, whereas offensive AI can reduce the cost of innovation $K_t$ or $c_t$. Since the model does not assume humans permanently dominate innovation, one can identify the parameters that determine the crossover.

\section{Algorithms}
Following the Stackelberg game, here we discuss algorithms for architecture design, as shown in Fig.~\ref{fig:algorithm}.

With $B$ candidate seams, the exact finite design problem enumerates $2^B$ subsets and solves \eqref{eq:convexprogram} for each. This is practical for standards-driven 6G architectural choices, where $B$ is often much smaller than the number of software components. For larger $B$, branch-and-bound, split--merge search, or mixed-integer convex approximations can be combined with the fixed-architecture oracle.

For example considering the seams indicated in Table~\ref{tab:seams}, $B=8$ and so one can evaluate 256 architectures. The attack graph is a DAG, so all source-to-target paths can also be enumerated exactly; depending on the seams, there are 148--182 paths. As we will see, for the architecture we end up selecting for a stylized O-RAN system, there will be 171 paths and eight control sites: four persistent external surfaces plus four selected boundaries.

\begin{figure}
\centering
\fbox{\begin{minipage}{0.94\columnwidth}
\footnotesize
\textbf{Exact architecture design for a finite seam set}
\begin{enumerate}
\item For each $\sigma\subseteq\mathcal{B}$, construct $G_A(\sigma)$, the control sites, $C_{\rm arch}(\sigma)$, and the path coefficients $(b_p,A_p)$.
\item Solve \eqref{eq:convexprogram}. For a large path set, begin with a small constraint set and use \eqref{eq:shortestpath} as a separation oracle until no path violates $z+A_px\geq b_p$.
\item Recover $\Theta=e^z$, $a^*=\Theta/c_A$, the dominant target/path, and total leader objective \eqref{eq:leaderproblem}.
\item Return the architecture and controls having smallest objective; retain other nondominated points as an architecture--security Pareto set.
\end{enumerate}
\end{minipage}}
\caption{Finite-seam Stackelberg design procedure. This  enumerates all architectures and all paths.}
\label{fig:algorithm}
\end{figure}

One can solve \eqref{eq:convexprogram} using sequential quadratic programming with analytic gradients and affine-constraint Jacobians. One can independently cross-check it by minimizing over $z$. For a fixed $z$, the problem
\begin{equation}
 \min_x\;c^Tx\quad
 \text{s.t. }A_px\geq b_p-z,\;0\leq x\leq\bar x
 \label{eq:LPcheck}
\end{equation}
 is a linear program. A bounded scalar minimization of its value plus $\beta e^{2z}$ therefore gives an independent convex solution. 
 
The best path may be calculated  by either direct path enumeration or by shortest paths. 

\section{Stylized O-RAN/6G Design}

Consider the stylized O-RAN system shown in Fig.~\ref{fig:arch} that has the components that were listed in Sec.~\ref{sec:oran}, together with the possible seams for partitioning shown in Table~\ref{tab:seams}.  Here we demonstrate the strategic design of modular architectures.

\begin{figure*}
\centering
\includegraphics[width=0.9\textwidth]{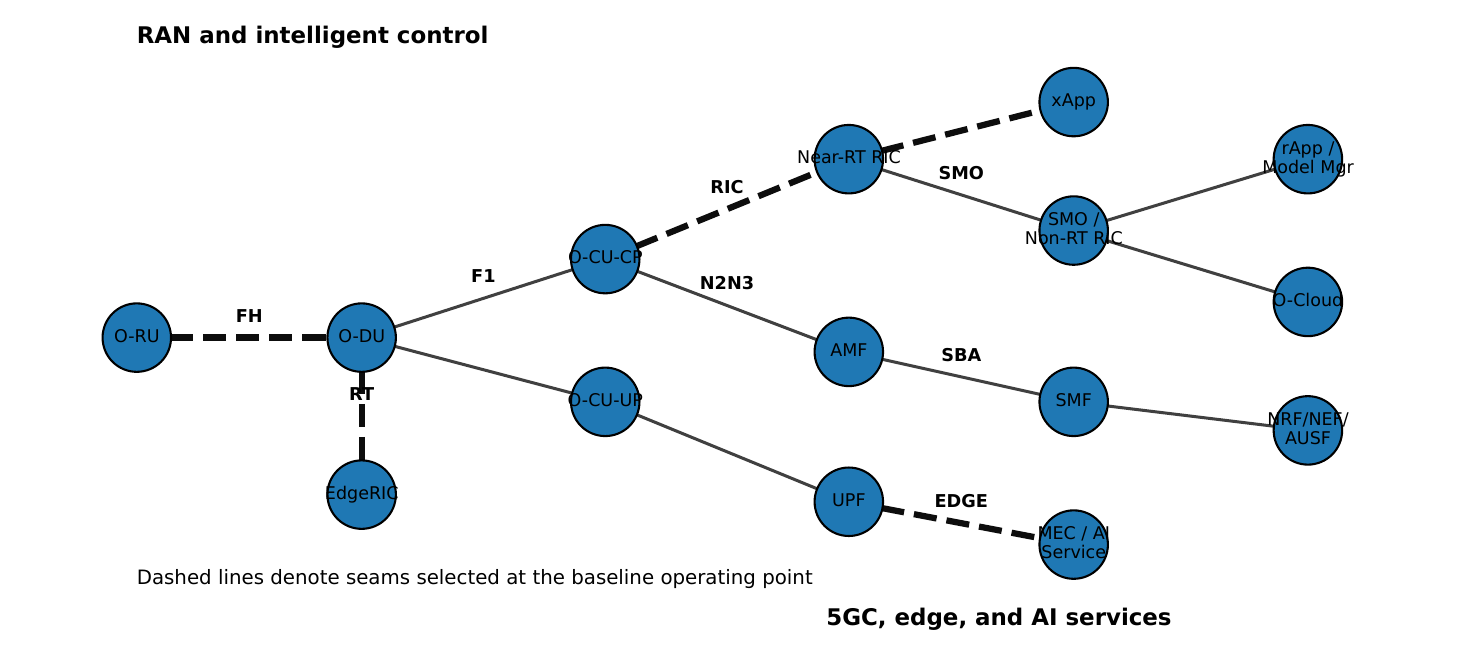}
\caption{Standards-anchored RAN--core--edge system. Dashed edges are seams selected at the baseline Stackelberg optimum. A seam may cut several deployment edges, so the number of modules need not equal $|\sigma|+1$. This is a stylized example rather than a proposed standards diagram.}
\label{fig:arch}
\end{figure*}

\subsection{Scenario}
Let the stylized example emphasize AI-RAN control, in terms of consequences. Mission consequences are set highest for EdgeRIC and near-RT RIC (cost $18$), followed by O-DU (cost $15$), O-CU-CP (cost $14$), O-CU-UP, SMF, and UPF (cost $12$), and AMF (cost $11$). The xApp, model, cloud, exposure, and edge-application nodes have lower direct values but can enable paths to high-value control and forwarding functions. These are normalized consequence weights rather than empirical compromise frequencies.

The base attack graph has 29 directed edges. External routes include radio/site access, O-Cloud or supply-chain access, NEF/NRF/AUSF exposure, MEC/AI application access, and embedded application/model-supply paths. Lateral edges represent, for example, O-Cloud to SMO and DU, rApp/model manager to SMO or xApp, xApp to near-RT RIC, near-RT RIC over E2 toward DU/CU, DU over F1, CU over N2/N3, and core service or N4 transitions. Selecting a seam raises specified crossing-edge exposure by $8$--$16\%$ and creates one or more new entry edges. The values are chosen to generate qualitatively realistic regimes and were not estimated from incident data.

The baseline normalized parameters are
\begin{equation}
 \lambda=0.4,\quad c_A=1,\quad c_M=0.005,
\end{equation}
with seam coordination costs scaled by $0.05$. Hardening and monitoring costs use site-specific coefficients, and detector exponents are largest for real-time RIC, E2/xApp, service-based core, SMO, and edge-AI boundaries, reflecting their potential for rich software and telemetry instrumentation. The control bounds are finite to prevent treating any interface as perfectly secure.

EdgeRIC exposes the architectural tradeoff sharply. It decouples real-time AI control from the RAN stack while remaining close enough to the DU for sub-millisecond operation \cite{ko2024edgeric}. The model  treats the real-time controller as a candidate boundary with high control consequence, added supply/API exposure, and strong monitoring effectiveness. Similarly, the edge-AI service is motivated by cloud-native 6G traffic intelligence and visibility services \cite{zeb2021edge}.

\begin{table}
\caption{Stylized Normalized Parameters}
\label{tab:parameters}
\centering
\small
\renewcommand{\arraystretch}{1.08}
\begin{tabular}{@{}ll@{}}
\toprule
Quantity & Baseline value or range\\
\midrule
Candidate seams / architectures & $8$ / $256$\\
Cyber weight $\lambda$ / effort cost $c_A$ & $0.4$ / $1.0$\\
Module cost $c_M$ & $0.005$\\
Coordination-cost scale & $0.05$\\
Hardening / monitoring cost scales & $0.6$ / $1.0$\\
Seam exposure multipliers & $1.08$--$1.16$\\
Seam entry success factors & $0.14$--$0.23$\\
Detector-effectiveness coefficients & $0.55$--$1.40$\\
Target consequence weights & $2$--$18$\\
\bottomrule
\end{tabular}
\end{table}

\subsection{Sensitivity and attacker population}
The baseline architecture search is performed over all 256 seam subsets. A two-dimensional sensitivity study reoptimizes controls over a set of 34 architectures: the 32 lowest baseline objectives plus the best architecture at every module count, the monolith, and full disaggregation. This reduced set is enough since a full enumeration is unnecessary to demonstrate phase transitions we find.

For the attacker-population illustration, ten types have $c_t=1$ and fixed-cost thresholds equivalent to yield cutoffs from $0.70$ to $2.10$. The bypass yield is set to $1.65$. 

\subsection{Results}
Let us report results of the stylized O-RAN system design.

\subsubsection{An interior architecture is selected}
Table~\ref{tab:mainresults} summarizes the principal operating points. At the stated normalized costs, the exact search selects seams FH, RT, RIC, and EDGE, yielding six deployment modules. This intermediate architecture has lower total objective and lower equilibrium cyber loss than the optimized monolith; the table reports the corresponding values. Its maximum yield and attacker effort are both $0.9504$.

One should not, however, interpret the result as \emph{more modules are safer}. Full disaggregation has lower cyber loss than the monolith but a higher total objective ($1.14423$) due to architecture and control costs. Fig.~\ref{fig:objective} shows the nonmonotone best objective by module count. The optimum occurs at an intermediate granularity, and several module counts have the same equilibrium cyber loss but different non-cyber costs.

\begin{table}
\caption{Stackelberg Operating Points}
\label{tab:mainresults}
\centering
\small
\renewcommand{\arraystretch}{1.10}
\begin{tabular}{@{}lrrr@{}}
\toprule
Design & $K$ & Cyber loss & Total objective\\
\midrule
Optimized monolith & 1 & 2.097 & 1.059\\
Selected seams & 6 & 0.903 & \textbf{0.937}\\
Full disaggregation & 15 & 0.972 & 1.144\\
\bottomrule
\end{tabular}
\end{table}

\begin{figure}
\centering
\includegraphics[width=\columnwidth]{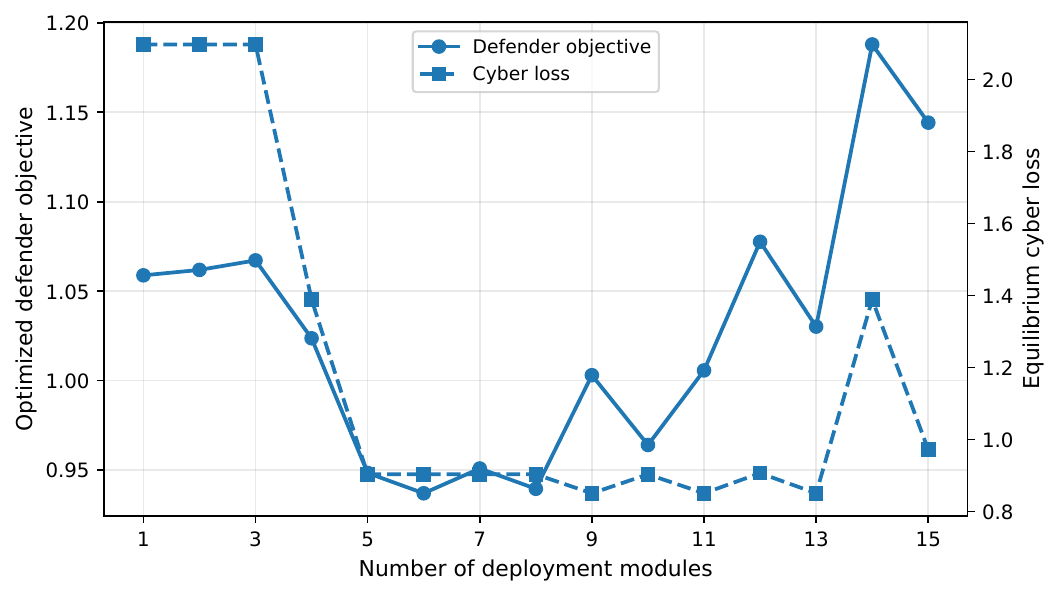}
\caption{Best objective and equilibrium cyber loss by module count. The architecture decision is nonmonotone; $K=6$ is the baseline operating point.}
\label{fig:objective}
\end{figure}

The optimizer allocates hardening primarily to persistent external routes (radio/site, O-Cloud, exposure/API, and MEC/AI entry) and monitoring primarily to the selected FH, RT, RIC, and EDGE boundaries. This division is explained by \eqref{eq:kkt}: external entry sites benefit from reducing baseline traversal, while the selected seams offer strong detector exponents across several co-dominant paths. At the optimum, several routes tie at $\Theta=0.9504$; examples include an exposure-service route to SMF and a route to the O-DU. Such equalization is characteristic of minimax protection.

\subsubsection{Bare disaggregation is harmful}
Consider the ablation study reported in Fig.~\ref{fig:ablation}, where various aspects of the design are omitted. With the selected seams but no hardening or monitoring, cyber loss is $16.46$, compared with $12.53$ for the bare monolith: the partition raises risk by 31.3\%. Once the defender reoptimizes, hardening alone lowers cyber loss to $1.0745$, monitoring alone to $1.3124$, and the joint design to $0.9033$. Monitoring alone does not beat the optimized monolith in total objective, while hardening alone does so only slightly. The large gain requires both controls: new interfaces are not safe merely because they are visible, and hardening alone does not exploit the additional telemetry made available by explicit boundaries.

\begin{figure}[t]
\centering
\includegraphics[width=0.98\columnwidth]{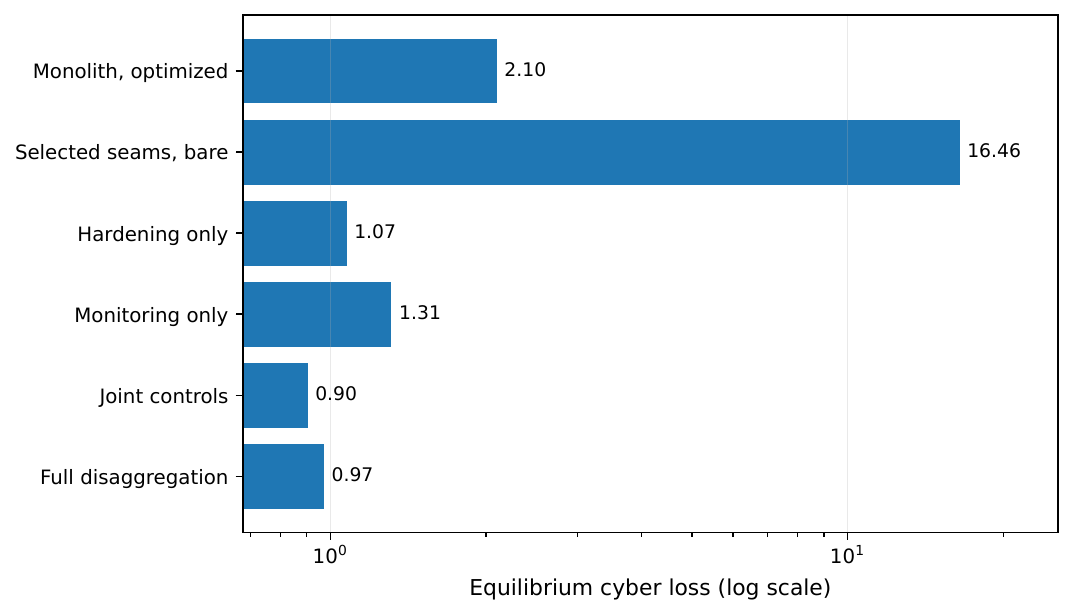}
\caption{Ablation of the selected architecture (logarithmic horizontal scale). Explicit seams without controls increase attack-path yield; improvement requires using both monitoring and hardening.}
\label{fig:ablation}
\end{figure}

\subsubsection{AI visibility versus AI attack surface}
Fig.~\ref{fig:phase} varies two quantities with direct AI-cyber interpretations: the detector-effectiveness scale multiplying $I_j$ and the interface-exposure scale multiplying newly exposed traversal and entry probabilities. When exposure is high and detector effectiveness is modest, the monolith is selected. For example, exposure scale $2.5$ with detector scale $0.5$, $0.75$, or $1.0$ yields $K=1$. When detector effectiveness is high and added exposure is low, the design supports eight to eleven modules. The boundary is the empirical counterpart of conditions \eqref{eq:AIthreshold} and \eqref{eq:localcriterion}.

\begin{figure}[t]
\centering
\includegraphics[width=0.98\columnwidth]{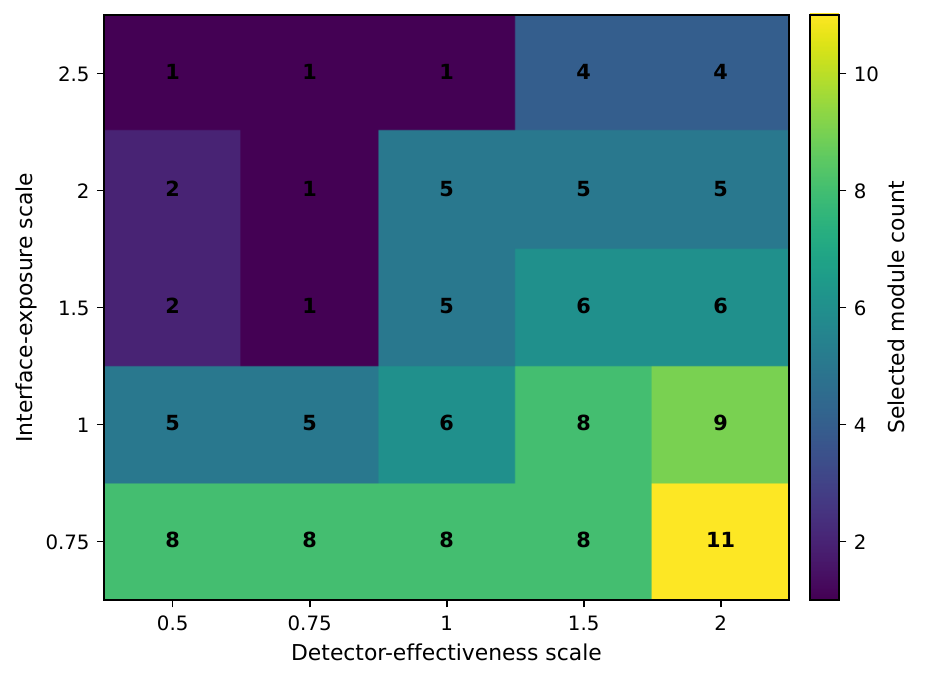}
\caption{Selected module count as detector effectiveness and interface exposure vary. The baseline is $(1,1)$; each cell reoptimizes over the listed candidate set.}
\label{fig:phase}
\end{figure}

This comparative static has two implications. First, automated code assurance and patching that reduce $\rho_e^0$ make more disaggregation supportable. Second, AI monitoring that raises $I_j$ can shift the optimum toward explicit modules, but only if the transducer's own service and model surface do not raise $\alpha_b$ faster. O-RAN's recent emphasis on AI security requirements, software authenticity, model-dependency checks, and continuous monitoring is consistent with treating these as coupled design variables rather than separate checklists \cite{oran2026security,soltani2025intelligent}.

\subsubsection{Attacker entry and innovation}
Fig.~\ref{fig:population} applies \eqref{eq:typepayoff} and \eqref{eq:bypass}, with respect to attacker innovation. At detector scale $0.4$, four of ten illustrative attacker types are profitable; the optimized architecture adapts as detector effectiveness increases, and no type remains profitable at scale $2.0$. The relative willingness to pay for the illustrated bypass rises from $0.824$ to $1.124$. The two effects are not contradictory. Better defense suppresses commodity paths and narrows entry, but a remaining actor that can create a sufficiently valuable bypass has more strategic reason to do so.

\begin{figure}[t]
\centering
\includegraphics[width=0.98\columnwidth]{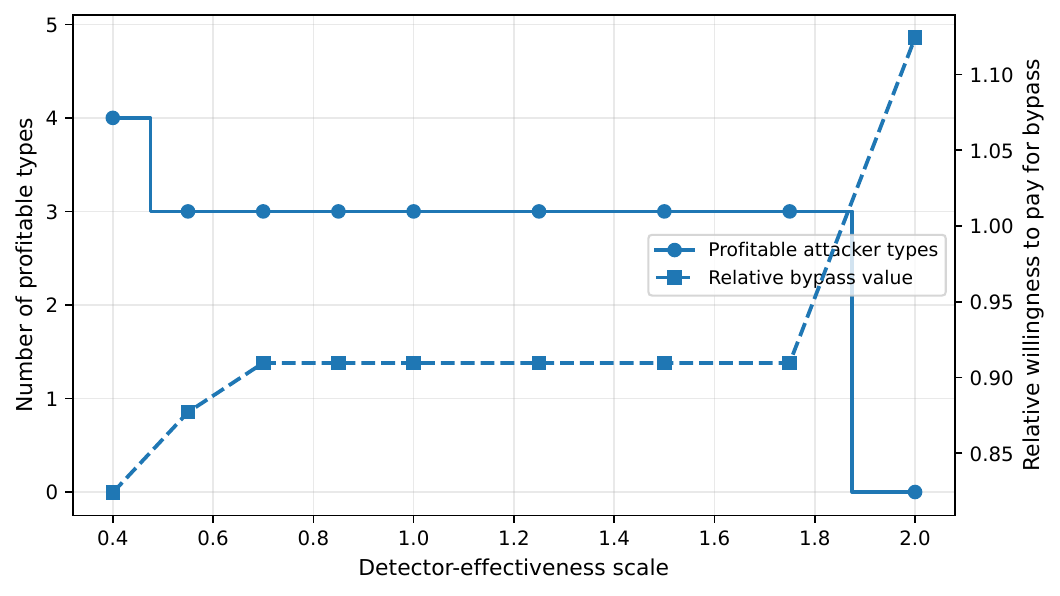}
\caption{Illustrative attacker-entry and bypass comparative statics. Improved detection contracts the profitable ordinary-attacker set while increasing the relative value of a novel bypass.}
\label{fig:population}
\end{figure}

This result gives a qualified mathematical basis for shifting defensive testing from only high-volume known attacks toward rare topology-aware bypasses as ordinary interfaces become well patched and monitored.  Offensive AI can lower bypass development cost; defensive AI can reduce baseline exploitability and improve detection. The relevant policy variable is the relative movement of $\rho_e^0$, $I_j$, $K_t$, and $c_t$.

\section{Extension: Global and Path-Dependent Attack Saturation}
\label{sec:saturation}
Insights thus far inspire an extension.
The baseline model makes attacker effort endogenous, but the gross return
$a\Theta_p$ remains linear in effort. This is appropriate when effort
represents scalable campaign volume. When attack effectiveness is constrained by a finite campaign capacity or by a finite set of target-specific opportunities, however, marginal return may decline. We
distinguish a \emph{global} saturation model, whose effort transformation is
common to every path, from a \emph{path-dependent} model, whose saturation
rate also reflects path survival. The first preserves the shortest-path response whereas the second can change the global target ranking. 

The two saturation models capture different realistic sources of diminishing attack returns. Global saturation represents campaign-level limitations that affect all attack paths similarly, such as finite attack windows, rate limits, limited compute or operator capacity, and diminishing value from repeatedly exploiting the same vulnerability. Path-dependent saturation instead captures multistage attacks in which raw effort must successfully traverse a particular route before reaching a finite set of opportunities at the target, so paths with different survival probabilities saturate at different rates.

\subsubsection{Global saturation preserves the shortest-path response}
Define the common effective-effort function
\begin{equation}
 s_{\bar a}(a)=\bar a\left(1-e^{-a/\bar a}\right),\qquad \bar a>0,
 \label{eq:global-saturation-function-revised}
\end{equation}
and let
\begin{align}
 U_A^{\rm global}(p,a;d)
 &=\Theta_p(d)s_{\bar a}(a)-\tfrac{c_A}{2}a^2,
 \label{eq:global-saturation-utility-revised}\\
 L^{\rm global}(p,a;d)
 &=\Theta_p(d)s_{\bar a}(a).
 \label{eq:global-saturation-loss-revised}
\end{align}
The normalization $s'_{\bar a}(0)=1$ means that small effort has the same
initial marginal return as in the baseline model, while
$s_{\bar a}(a)\leq\bar a$ makes the total effective effort bounded.

\begin{proposition}[Global saturation and path choice]
For fixed $d$, every positive-effort best response under
\eqref{eq:global-saturation-utility-revised} selects a path in
$\argmax_p\Theta_p(d)$. Hence Theorem~\ref{thm:shortestpath} is unchanged.
\end{proposition}
\begin{IEEEproof}
For fixed $a>0$, $s_{\bar a}(a)>0$ is common to all paths and the effort cost
is path-independent. Path ranking is therefore exactly the ranking by
$\Theta_p(d)$. At $a=0$ all paths tie, so a maximum-$\Theta_p$ path remains a
best response.
\end{IEEEproof}
As $\bar a$ grows, $s_{\bar a}(a)$ approaches $a$, so the model approaches the
baseline linear-return case.

\subsubsection{Path-dependent saturation changes the global target rule}
Define the one-attempt survival
of path $p$:
\begin{equation}
 R_p(d)=\prod_{e\in p}r_e(d)
 =\frac{\Theta_p(d)}{V_{t(p)}}.
 \label{eq:path-saturation-survival-simple}
\end{equation}
Let $\omega>0$ be a common
opportunity scale. Define
\begin{align}
 U_A^{\rm path}(p,a;d)
 &=\omega V_{t(p)}
 \left(1-e^{-aR_p(d)/\omega}\right)-\tfrac{c_A}{2}a^2,
 \label{eq:path-saturation-utility-simple}\\
 L^{\rm path}(p,a;d)
 &=\omega V_{t(p)}
 \left(1-e^{-aR_p(d)/\omega}\right).
 \label{eq:path-saturation-loss-simple}
\end{align}
Initial marginal return is $V_{t(p)}R_p=\Theta_p$, while limiting return is
$\omega V_{t(p)}$. Thus initial attractiveness and ultimate consequence matter
separately; as $\omega\to\infty$, the baseline return $a\Theta_p$ is recovered.

For each reachable target $v$, let $R_v(d)$ be the maximum
one-attempt survival probability over all paths from $s$ to $v$.

The optimized payoff associated with target $v$ is
\begin{equation}
 \Psi_v(d)
 =\max_{a\geq0}
 \left\{
 \omega V_v\left(1-e^{-aR_v(d)/\omega}\right)-\tfrac{c_A}{2}a^2
 \right\}.
 \label{eq:path-saturation-target-value}
\end{equation}

\begin{proposition}[Path-dependent saturating response]
For every fixed target $v$ with a positive-effort response, the attacker uses
a shortest path to $v$. Across targets, however, it selects
\begin{equation}
 v^*\in\argmax_v\Psi_v(d),
 \label{eq:path-saturation-target-choice-simple}
\end{equation}
which need not be a target maximizing $\Theta_v=V_vR_v$.
\end{proposition}
\begin{IEEEproof}
For fixed $v$ and $a>0$, the gross-return term in
\eqref{eq:path-saturation-utility-simple} is strictly increasing in $R_p$.
Optimizing over effort preserves this ordering. The best route to $v$
therefore maximizes $R_p$, or equivalently minimizes
$-\log R_p=\sum_{e\in p}\ell_e$, and is a shortest path. Across targets,
the score is not a function of $\Theta_v=V_vR_v$ alone. Indeed, hold
$\Theta_v$ and $a>0$ fixed and substitute $R_v=\Theta_v/V_v$ into the gross
return. Its derivative with respect to $V_v$ is
$\omega[1-(1+x)e^{-x}]>0$, where $x=a\Theta_v/(\omega V_v)>0$.
Thus, whenever the lower-$V_v$ target's optimizer is positive, a
higher-$V_v$ target has a strictly larger optimized score at the same
$\Theta_v$; by continuity, it can also outrank a target with slightly larger
$\Theta_v$.
\end{IEEEproof}

Thus shortest paths survive within targets, but the global maximum-$\Theta_p$
rule need not. 

\subsubsection{Numerical comparison}
\textit{Saturation changes architecture choice and can invalidate the maximum-yield target rule:}
We first consider common exponential saturation, for which Proposition~6 guarantees that the attacker still selects a maximum-$\Theta_p$ path. To compare saturation strengths on a common scale, normalize the saturation parameter by the baseline linear-return effort,
$\eta=\frac{\bar a}{a_{\mathrm{ref}}^{\mathrm{lin}}}$, where $a_{\mathrm{ref}}^{\mathrm{lin}}=\frac{\Theta_{\mathrm{ref}}}{c_A}=0.9504$.
Small $\eta$ corresponds to strong saturation, whereas large $\eta$ approaches the baseline linear-return model. At each of 25 logarithmically spaced values
$\eta \in [10^{-2},10^2]$, we enumerate all 256 architectures and reoptimize hardening and monitoring. Strong saturation favors integration: the monolith is selected through the sampled value $\eta=6.8129$. As saturation weakens, the original
$\{\mathrm{FH},\mathrm{RT},\mathrm{RIC},\mathrm{EDGE}\}$
six-module architecture reemerges and is selected from $\eta=10$ onward. Thus saturation can shift the preferred architecture even when the attacker's shortest-path ranking remains unchanged.

Path-dependent saturation has a qualitatively different effect: it can change which target the attacker prefers. To compare with the common-saturation model, let
$R_{\mathrm{ref}}=0.0528$
be the baseline survival of the canonical reference path and set
\begin{equation}
\omega
=
\eta a_{\mathrm{ref}}^{\mathrm{lin}} R_{\mathrm{ref}}.
\tag{45}
\end{equation}
A path with survival $R_{\mathrm{ref}}$ then has effective scale
$\omega/R_{\mathrm{ref}}
=
\eta a_{\mathrm{ref}}^{\mathrm{lin}}$,
matching the common normalization. For each $\eta$, freeze the 256 policies optimized under common saturation and recompute their exact path-dependent follower responses.

\begin{figure}[t]
    \centering
    \includegraphics[width=\columnwidth]
    {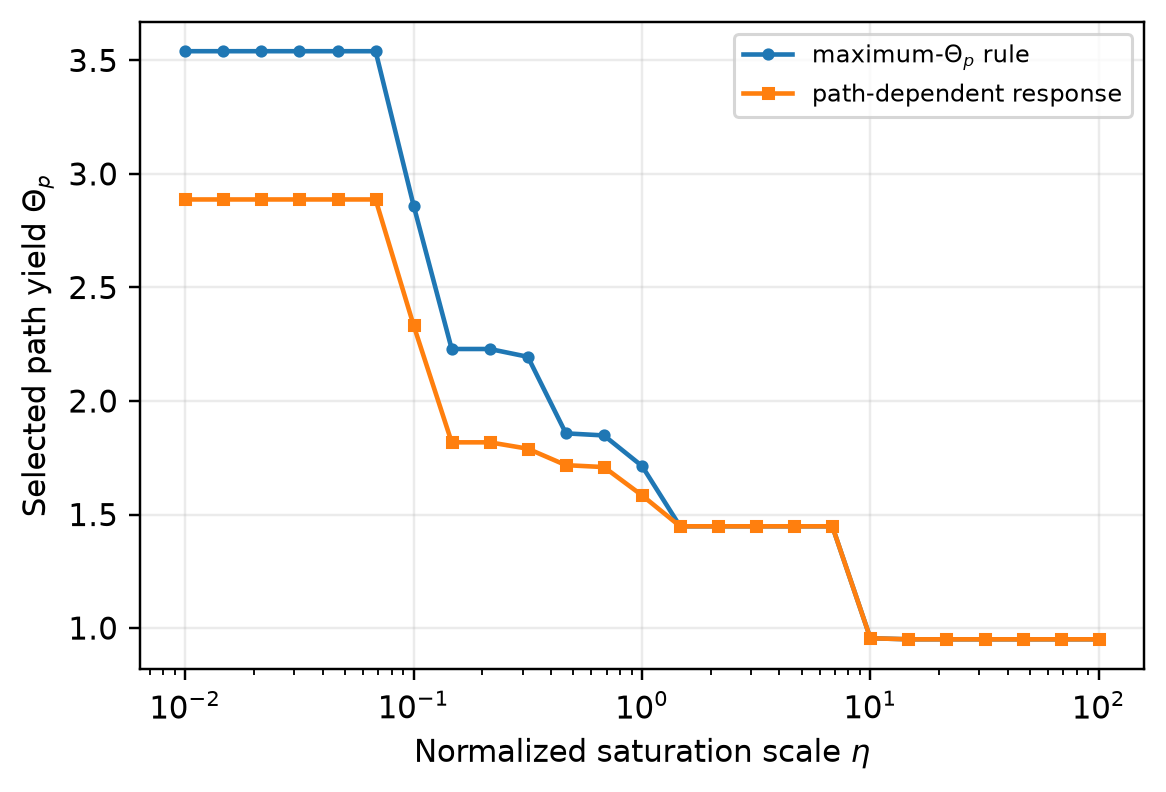}
    \caption{Selected path yield under the maximum-$\Theta_p$ rule
    and the path-dependent saturation response.}
    \label{fig:path-dependent-selected-yield}
\end{figure}

Fig.~8 shows that the baseline maximum-$\Theta_p$ target rule fails nontrivially: the path-dependent attacker selects a path with strictly smaller $\Theta_p$ at 13 of the 25 parameter values. Under the deterministic canonical tie convention, the selected target label changes at 18 values, with the additional five changes arising from ties among paths having the same maximum $\Theta_p$.

Experiments show numerous  target switches, illustrating the mechanism in Proposition~7. At $\eta=0.1$, for example, the maximum-$\Theta_p$ rule ranks an O-DU path with $\Theta=2.8602$ first, whereas the path-dependent attacker selects EdgeRIC with the smaller yield $\Theta=2.3327$, because EdgeRIC has the larger mission consequence, $18$ rather than $15$. At $\eta=1$, the maximum-$\Theta_p$ rule selects a UPF path with $\Theta=1.7136$, whereas the path-dependent response selects a near-RT RIC path with $\Theta=1.5844$. Hence initial path attractiveness and ultimate target consequence can matter separately once repeated attack opportunities saturate.

These target changes can also affect architectural evaluation. Among the frozen policies, recomputing the path-dependent follower response changes the preferred policy at several intermediate values of $\eta$, with a maximum defender-objective improvement of $0.04042$ relative to blindly retaining the common-saturation winner. This calculation is a policy-transfer stress test rather than a globally reoptimized solution of the path-dependent leader problem. As saturation weakens, both the target ranking and the preferred six-module design converge toward the baseline model, as expected.

\section{Design and Policy Implications}

Let us distill design and policy insights that emerge.

\subsubsection{Open and secure require interface-by-interface accounting}
The objective of open, interoperable, secure, and resilient 6G \cite{state2024principles,ntia2026call} should not be reduced to either maximal openness or maximal integration. Equation~\eqref{eq:exactcriterion} implies that each proposed interface must be evaluated through four quantities: added coordination cost, added attack paths, containment/hardening effectiveness, and detection effectiveness. A standard interface can be strategically valuable when it creates an auditable, least-privilege choke point; the same interface can be harmful when it merely makes a powerful function remotely reachable.

\subsubsection{AI-enabled secure development changes the architecture frontier}
Secure software development reduces baseline edge factors $\rho_e^0$, whereas automated detection and patching can shorten the duration for which an exposed edge remains attractive. NIST's SSDF provides a general secure-development framework for reducing software vulnerabilities \cite{nist2022ssdf}. In the game formulation, such improvements lower the exposure multiplier $\alpha_b$ or increase effective hardening $\gamma_bh_b$. They therefore shift the feasible frontier of modularity: verified or rapidly patched interfaces permit more decomposition before added surface dominates.

This does not imply that AI automatically favors disaggregation. AI controllers, xApps, rApps, model stores, feature pipelines, and inference services can create high-leverage source edges. The transducer condition \eqref{eq:AIthreshold} is the appropriate test. Procurement and standardization should measure both sides of the same AI component: its detector gain and its incremental exploitability.

\subsubsection{Zero trust is a resource allocation problem}
Zero trust emphasizes explicit authentication, authorization, least privilege, and continuous monitoring rather than implicit trust based on network location \cite{nist2020zta,oran2024zta}. In the model these mechanisms reduce traversal factors and raise detection exponents on particular edges. Equation~\eqref{eq:kkt} adds a prioritization rule: invest first on controls that intersect several co-dominant high-consequence paths. Uniform hardening of every interface can waste resources on routes that a strategic attacker would not select.

\subsubsection{From many ordinary actors to fewer innovation-capable actors}
Proposition~5 is a comparative-static result, not a prediction. Yet, we find that if secure development, monitoring, and containment lower $\Theta$, attacker types with small margins exit. If attribution raises $A_t$, the set contracts further. A small remaining set is easier to analyze strategically, but stability additionally requires credible attribution, response, and control of false flags, none of which are addressed here. The bypass condition \eqref{eq:bypass} nonetheless clarifies why offensive cyber portfolios focused only on attack volume can become less relevant against mature architectures: residual advantage comes from changing the attack graph, not merely repeating the best known path.

\section{Conclusion}
This paper provides a mathematical formulation of the 6G monolithic--modular security tradeoff that is core to trustworthy system design. Architectural disaggregation changes the attacker's path set, interface hardening and monitoring change the value of each path and a Stackelberg attacker chooses both the most attractive route and an endogenous effort level. The logarithmic path representation produces a shortest-path best response, while the defender's fixed-architecture problem emerges as convex. The resulting criterion shows that the visibility and containment benefits of a boundary must exceed its added attack opportunity and coordination cost.

The stylized O-RAN design provides further insight into the theoretical result. At a baseline operating point, the best design is neither the monolith nor full disaggregation. The intermediate architecture selected in the example improves the optimized objective, yet the same seams without security controls are substantially worse than the bare monolith. As AI detector effectiveness improves relative to AI-enabled interface exposure, greater disaggregation is appropriate. As ordinary path yield falls, the profitable attacker population contracts and residual incentives move toward novel bypasses. 

Architectural disaggregation is only one component of trustworthy 6G. End-to-end trust also depends on numerous factors including privacy, resilience, secure software and hardware supply chains, robust AI, and protection against jamming, spoofing, and malicious sensing. Yet architecture remains foundational, especially with growing softwarization.  As 6G integrates communication, sensing, computation, control, and autonomous intelligence, interface placement determines where trust must be established, where attacks can be observed, and how failures propagate. The present framework complements other dimensions of trustworthy 6G by helping choose the system structure within which they must operate coherently.

To push the mathematical approach to operational use, one would need to determine real-world parameters drawing on evidence from red-team exercises, vulnerability databases, software inventories, O-RAN telemetry, detector operating curves, incident-response data, and service-impact models.  The model itself can be extended beyond a DAG attack graph with multiplicative conditional factors, since real attacks can revisit states, exploit correlated credentials, and adapt after partial detection.  To handle such cyclic graphs, the shortest-path form may need modification.
The continuous detector model abstracts false-alarm and operational-load costs into $c_n$ and $I_j$; a full design could directly draw on hypothesis testing results and optimize thresholds using measured receiver-operating characteristics.

\bibliographystyle{IEEEtran}
\bibliography{references}

\end{document}